\documentclass[journal,twoside,web]{ieeecolor}
\usepackage{generic,url}
\usepackage{cite}
 \usepackage{balance}
 \usepackage{algorithmic}
 \usepackage{graphicx,epsfig,float}
\usepackage{amsmath,amssymb}
\usepackage{amsfonts}
\usepackage{latexsym}
\usepackage{calc}
\usepackage{tikz}
\usepackage{hyperref}   

\usepackage{mdframed}
\usepackage{lipsum}

\hypersetup{
  verbose,
  pdfnewwindow,
  plainpages=false,
  bookmarksnumbered={true},
  colorlinks={true},
  linkcolor={blue},
  anchorcolor={blue},
  citecolor={blue},
  filecolor={blue},
  menucolor={red},
  urlcolor={blue}
}

\newcount\m
\newcount\n

\def\twodigits#1{\ifnum #1<10 0\fi \number#1}

\def\hours{\n=\time \divide\n 60
    \m=-\n \multiply\m 60 \advance\m \time
    \twodigits\n:\twodigits\m}

\def\data{Rio de Janeiro,\  \number\day\  de \ifcase\month\or
    janeiro\or
    fevereiro\or
    mar\c{c}o\or
    abril\or
    maio\or
    junho\or
    julho\or
    agosto\or
    setembro\or
    outubro\or
    novembro\or
    dezembro\or\fi\  de \number\year}

\definecolor{gG}{RGB}{ 60, 186,  84}  
\definecolor{gY}{RGB}{244, 194,  13}
\definecolor{gB}{RGB}{ 72, 133, 237}
\definecolor{gR}{RGB}{219,  50,  54}

\definecolor{myG}{RGB}{ 10, 220, 10}  

\def\D{\mathbb{D}}

\def\R{\mathbb{R}}

\def\diag{\mbox{\textrm{\,diag}}}

\def\Ltwo{\mbox{${\mathcal L}_2$}}
\def\Linf{\mbox{${\mathcal L}_\infty$}}

\def\sign{\mbox{\textrm{\,sign}}}

\input{formato}

\newtheorem{lem}{\textbf{Lemma}}
\newtheorem{rem}{\textbf{Remark}}

\newmdtheoremenv{theo}{Theorem}

\newcommand{\DOI}[1]{%
  \href{http://dx.doi.org/#1}{$<$doi$>$}
}

\def\UP{\rule[-1mm]{0mm}{5mm}}
\def\IN{\!\in\!}

\def\Lplus{L_{+}}
\def\dmin{d^*}

\def\A{\mathcal{A}}

\def\Ups{\Upsilon}
\def\O{\mathcal{O}}
\begin{document}

 \def\BibTeX{{\rm B\kern-.05em{\sc i\kern-.025em b}\kern-.08em
 		T\kern-.1667em\lower.7ex\hbox{E}\kern-.125emX}}
 \markboth{\journalname, VOL. XX, NO. XX, XXXX 2024}
 {Author \MakeLowercase{\textit{et al.}}: Preparation of Papers for IEEE TRANSACTIONS and JOURNALS (xxxx 2024)}

	\title{Arbitrarily Fast Multivariable Least-squares MRAC}
	 \author{Liu Hsu, \IEEEmembership{Senior Member, IEEE}, 
     Ramon R. Costa, 
     Fernando Lizarralde, \IEEEmembership{Senior Member, IEEE} and Alessandro J. Peixoto,
		\thanks{This work was supported in part by the National Conselho Nacional de Desenvolvimento Científico e Tecnológico (CNPq / Brazil), under Grant 302013/2019-9, FAPERJ, under Grant E-26/204.669/2024  and Coordenação de Aperfeiçoamento de Pessoal de Ensino Superior (Capes / Brazil), Finance code 001.}
        \thanks{L. Hsu, R. R. Costa, F. Lizarralde and A. J. Peixoto are with the Department of Electrical Engineering, COPPE, Federal University of Rio de Janeiro, RJ-Brazil, (email: lhsu@coppe.ufrj.br).}}
	
	\maketitle

\author{}
\thanks{}

\maketitle

\begin{abstract}
		A \textcolor{black}{novel} least-squares model-reference direct adaptive control (LS-MRAC) algorithm for multivariable (MIMO) plants is presented. The controller parameters are directly updated based on the \textcolor{black}{output} tracking error. The control law is \textcolor{black}{crucially} modified to reduce the relative degree of the error model to zero. A complete Lyapunov-based stability analysis as well as a tracking error convergence characterization is provided demonstrating that the LS-MRAC can achieve arbitrarily fast tracking while maintaining \textcolor{black}{satisfactory}  parameter convergence for appropriate adaptation gains. 
        Simulation results show a significant improvement in tracking performance compared to previous methods.
\end{abstract}

\begin{IEEEkeywords}
Adaptive control, model-reference adaptive control, stability, transient performance, least-squares.
\end{IEEEkeywords}

\section{Introduction}
Model Reference Adaptive Control (MRAC) is one of the primary approaches in adaptive control. Over the past decades, extensive research has been conducted on designing MRAC systems with proven stability and convergence properties \cite{IS:96}\cite{SB:89}\cite{Tao:2003}\cite{nguyen2018}. 

Two common alternatives for adaptation are the gradient and least-squares (LS) methods. While the former can be applied for direct or indirect adaptation, the latter has been used for indirect adaptation \cite{nguyen2018}. It is well known that LS can potentially lead to faster parameter estimation \cite{berghuis1995}\cite{krstic2009using}\cite{guler2016adaptive}, \cite{fidan2015least}. Therefore, it would be \textcolor{black}{highly interesting} 
to develop a direct \textcolor{black}{singularity-free} LS adaptive control algorithm, 
to compute the control law when calculated indirectly from the identified plant parameters based on the certainty equivalence principle. An early attempt was made by \cite{ECD:1985} , but stability and tracking issues remained unresolved \cite{Tao:2014}. In \cite{Slotine}, a composite adaptive control scheme was proposed, which integrates both tracking error and prediction error into the update law.

Recently, \cite{ZFK:2022} proposed a direct LS-MRAC scheme for relative degree one SISO plants. Unfortunately, the adaptation law requires an unavailable signal to be implemented, namely the instantaneous parameter estimation error.

In this paper, we present a solution for designing output feedback direct LS adaptive control for linear time-invariant (LTI) MIMO plants with uncertain high-frequency gain $K_p$, considering only plants with uniform relative degree one.

Besides the uncertainty of $K_p$, which is amenable by SDU factorization \cite{CHIK:03}, another challenge is how to avoid the need for output differentiation to implement the LS adaptation law. Here, we circumvent this problem with the strategy proposed in \cite{Costa:2024, Costa:2020}, which consists in introducing a modified control law derived from Monopoli’s multiplier. As a result, the output error equation becomes of relative degree zero. This property is crucial as it leads to an implementable direct LS adaptive law using output feedback for MIMO plants. 

A novelty of the new adaptive law is that it does not require normalization 
(\textcolor{black}{which tend to decrease the equivalent adaptation gain}) to ensure stability. We believe that the proposed LS-MRAC, with a complete Lyapunov-based stability analysis, is the first of its kind in the literature.

A further important contribution is that a detailed tracking error convergence analysis is presented showing that the LS-MRAC  can be designed to achieve arbitrarily fast tracking while keeping satisfactory parameter convergence for appropriate adaptation gains \textcolor{black}{under the usual persistency of excitation condition}. In this paper, only standard LS adaptation is considered. However, new LS learning methods could be explored in the future (see \cite{Pan_Shi_Ortega:2024}).

This note is organized as follow: In Sec.~\ref{MIMOcase}, the problem is stated and the assumptions are formulated. Then the controller parametrization based on SDU factorization  of the control input matrix is described in  (Sec.~\ref{sec:SDU}). The controller design follows in Sec.~{\ref{sec:design}. A fast converging LS-MRAC is described with its stability analysis in Sec.~{\ref{sec:fastLSMRAC}}. The remarkable convergence properties are characterized in  Sec.~\ref{sec:convergence}. Simulation examples are presented in Sec.~\ref{sec:simulations} to illustrate the significant improvement in the speed of adaptation compared to alternative adaptation laws.

\section{Multivariable MRAC}\label{MIMOcase}

The plant and the reference model are given by\footnote{\textcolor{black}{The symbol $s$ stands either as the differential operator or as the Laplace variable, according the context of time-domain or frequency-domain.}}
\begin{align}
  y(t) &= P(s)\,u(t)\,, \quad P(s) = K_p N(s) D^{-1}(s)\,, \label{MLS:y} \\[1mm]
  y_m(t) &= M(s)\,r\,, \label{MLS:ym}
\end{align}
where $P(s)$ and $M(s)$ are $m \!\times\! m$ transfer matrices, and $y, u, y_m, r \IN \R^m$.
The high frequency gain matrix $K_p$ and the coefficients of the $m \!\times\! m$ polynomial matrices $D(s)$ and $N(s)$ are unknown.

The objective is to track the reference model by designing the control $u$ so that $e_0=y(t)-y_m(t) \rightarrow 0$ asymptotically while all signals of the closed-loop system remain bounded (stability).

As usual, the reference signal $r(t)$ is piecewise continuous and uniformly bounded.

The following assumptions regarding $P(s)$ are required:
\begin{enumerate}
  \item the observability index $\nu$ of $P(s)$ is known;
  \item $P(s)$ has relative degree $1$;
  \item the transmission zeros of $P(s)$ have negative real parts 
 (minimum-phase assumption);   
 \item the signs of the leading principal minors $\{\Delta_i\}_{i=1}^m$ of the matrix $K_p$ are known; 
\end{enumerate}

Assumption 4) is a counterpart of a similar high-frequency gain sign assumption made for the SISO LS-MRAC.

In view of assumption 2), $M(s)$ is chosen as
\begin{align}\label{eq:M_Am}
  M(s) = (sI-A_m)^{-1}\,,%
\end{align}
where 
\begin{align} \label{eq:Am}
  A_m = \diag \big\{ -a_i \big\}_{i=1}^m\!, \quad a_{i}> 0\,,
\end{align}

The \emph{state variable filters} are given by
\begin{align}
  \dot{v}_{1,i} &= \Lambda v_{1,i} + g u_i\,, \quad v_{1,i} \IN \R^{\nu-1}\,, \label{MLS:dv1} \\
  \dot{v}_{2,i} &= \Lambda v_{2,i} + g y_i\,, \quad v_{2,i} \IN \R^{\nu-1}\,, \label{MLS:dv2} \\[2mm]
  v_1^T &= \begin{bmatrix} v_{1,1}^T & v_{1,2}^T & \cdots & v_{1,m}^T \end{bmatrix} \IN \R^{m(\nu-1)}\,, \notag \\
  v_2^T &= \begin{bmatrix} v_{2,1}^T & v_{2,2}^T & \cdots & v_{2,m}^T \end{bmatrix} \IN \R^{m(\nu-1)}\,, \notag
\end{align}
where $\Lambda$ is a Hurwitz matrix.
The regressor vector is set as
\begin{align*}
  \omega^T &= \begin{bmatrix} v_1^T & v_2^T & y^T & r^T \end{bmatrix} \in \R^{2m\nu}\,.
\end{align*}

If $P(s)$ is known, then a control law which achieves matching between the
closed-loop transfer matrix and $M(s)$, i.e. $y = P(s) u^* = M(s) r =
y_m$, is given by \cite{Tao:2003}
\begin{align}\label{MLS:u*}
  u^* = \theta^{*T}_1 v_1 + \theta^{*T}_2 v_2 + \theta^*_3 y + \theta^*_4 r
      = \theta^{*T} \omega\,,
\end{align}
where $\theta^{*T} = \begin{bmatrix} \theta_1^{*T} & \theta_2^{*T} & \theta_3^* & \theta_4^* \end{bmatrix}$,
  $\theta_1^*, \theta_2^* \IN \R^{m(\nu-1) \times m}$,
  $\theta_3^* \IN \R^{m \times m}$, and
  $\theta_4^* = K_p^{-1}$.
The dynamics of the tracking error is given by \cite[p. 201]{Tao:2003}
\begin{align}\label{MLS:e0}
  e_0 &= M(s) K_p \left[u - \theta^{*T}\omega \right]\,. 
\end{align}
%

\subsection{$SDU$ factorization}\label{sec:SDU}

In order to circumvent strong prior limitations about $K_p$ such as its symmetry, the MIMO MRAC algorithm proposed in \cite{CHIK:03} relies on a control parametrization derived from an $SDU$ (Symmetric, Diagonal, Upper Triangular) factorization of the matrix $K_p$.
The following Lemma is central for this algorithm.
It assures the existence of an $SDU$ factorization.

\begin{lem}\label{lem1}
Every $m \times m$ real matrix $K_p$ with nonzero leading principal minors $\{\Delta_i\}_{i=1}^m$
can be factored as
\begin{equation}\label{lem1:SDU}
  K_p = SDU,
\end{equation}
where $S$ is symmetric positive definite, $D$ is diagonal, and $U$ is unity upper triangular.
\end{lem}

\begin{proof} From \cite{CHIK:03}.

Since all $\Delta_i$ are nonzero, there exists a unique factorization \cite{Strang:88},
\begin{equation}\label{lem1:Kp}
  K_p = L_pD_p\,U_p\,,
\end{equation}
where $L_p$ is unity lower triangular, $U_p$ is unity upper triangular, and
\begin{equation}
  D_p = \diag\left\{\Delta_1,\frac{\Delta_2}{\Delta_1},\dots,
  \frac{\Delta_m}{\Delta_{m-1}}\right\}\,.
\end{equation}
%
%
Factoring $D_p$ as
\begin{equation}\label{lem1:Dp}
  D_p = D_+ D\,,
\end{equation}
where $D_+$ is a diagonal matrix with positive entries, (\ref{lem1:Kp}) is rewritten as
\begin{equation*}
  K_p = \big( L_p D_+\, L_p^T \big) D \big( D^{-1} L_p^{-T} D\, U_p \big)\,,
\end{equation*}
so that (\ref{lem1:SDU}) is satisfied by $D=D_p D^{-1}$ and 
\begin{align}
  S &= L_p\, D_+\, L_p^T\,, \label{lem1:S} \\
  U &= D^{-1} L_p^{-T} D\, U_p\,.
\end{align}
\end{proof}

Note that while the LDU factorization is unique the SDU factorization is not since $D_+>0$ is arbitrary. 

\subsection{Control parametrization}

Applying the $SDU$ factorization (\ref{lem1:SDU}), the error equation (\ref{MLS:e0}) is rewritten as
\begin{align}
  e_0 &= M(s) SDU [u -\theta^{*T} \omega] \notag \\
    &= M(s) SD\, [Uu - U\theta^{*T} \omega] \notag \\
    &= M(s) SD\, [u - U\theta^{*T} \omega - (I-U) u]\,. \label{MLS:e0:2}
\end{align}

Note that the matrix $(I-U)$ is strictly upper triangular.
Therefore, the control signal $u$ can be defined as a function of $(I-U)u$ without given rise to any static loop since no individual control signal $u_i$ would depend on itself.

The unknown $U$ is incorporated in the parametrization by introducing the parameter vector $\Theta^*$ via the identity
\begin{equation}
  \Theta^{*T}\Omega \equiv U \theta^{*T}\omega + (I-U)u\,,
\end{equation}
where
\begin{align}
  \Theta^{*T} &= \begin{bmatrix} \Theta_1^{*T} & \Theta_2^{*T} & \cdots & \Theta_m^{*T}\end{bmatrix}\,, \\[1mm]
  \Omega^T &= \diag\big\{ \Omega_i^T \big\}_{i=1}^{m}\,, \label{MLS:Omega}
\end{align}
and
\begin{align}\label{MLS:Omegai}
  \Omega_1^T &= [\omega^T \quad u_2 \quad u_3 \ \dots \ u_m], \notag \\
  \Omega_2^T &= [\omega^T \quad u_3 \ \dots \ u_m],\notag \\[-1.5mm]
  \vdots \notag \\[-1.5mm]
  \Omega_{m-1}^T &= [\omega^T \quad u_m],\notag \\
  \Omega_m^T &= [\omega^T]\,.
\end{align}

Each vector $\Theta_i^{*T}$ concatenates the $i$-th row of the matrix $U \theta^{*T}$ and also the non zero entries of the $i$-th row of $(I-U)$.
The error equation (\ref{MLS:e0:2}) has thus been brought to the form
\begin{align}\label{MLS:e0:3}
  e_0 &= M(s) S\,D\,[u - \Omega^T \Theta^*]\,.
\end{align}

where $\tilde u= u-\Omega \Theta^*$ is the control error or mismatch. Note that this parametrization introduces $m(m-1)/2$ additional adaptive parameters.

A special parametrization of the matrix $D_+$ is proposed in Lemma \ref{lemma3} below.
This result will be useful to conclude the stability of the MIMO extension of the LS-MRAC algorithm developed in the next in terms of the unique LDU factorization.

\textcolor{black}{\begin{lem}\label{lemma3}
%
For any $A \!=\! \diag\big\{ -a_i \big\}_{i=1}^m$, $a_{i} \!>\! 0$,
and $K_p$ satisfying Lemma \ref{lem1}, there exists a factorization SDU such that 
%
%
\begin{align*}
  2Q = -(AS^{-1} + S^{-1}A) > 0\,. 
\end{align*}
\end{lem}
}
\begin{proof} 

\textcolor{black}{Let us choose 
$$D_{+}=diag\{ 1,\ d_+^2,\ d_+^4, \cdots, d_+^{2(m-1)}\}, \quad d_+>d^*\,, $$ }

where $d^*>0$ is sufficiently large. Using the factorization $D_{+} = D_{+}^{1/2} D_{+}^{1/2}$, we have with (\ref{lem1:S})
\begin{align}
  2 Q &= 
  -S^{-1} D_{+}^{1/2} \big( \Lplus \Lplus^T A + A \Lplus \Lplus^T \big) D_{+}^{1/2} S^{-1}\,. \label{lem3:2Q}
\end{align}
%
%
%
%

Let us choose $\Lplus = D_{+}^{-1/2}L_p D_{+}^{1/2}$. 
Denoting  $L_p=(l_{ij}), \ i,j=1,\ldots m$, we get 
\begin{align} \label{MLS:L+}
  \Lplus = \begin{bmatrix}
  1 & 0 & 0 & 0 & \cdots & 0 \\[2mm]
  \dfrac{\ell_{21}}{d_+} & 1 & 0 & 0 & \cdots & 0 \\[2mm]
  \dfrac{\ell_{31}}{d_+^2} & \dfrac{\ell_{32}}{d_+} & 1 & 0 & \cdots & 0 \\
  \dfrac{\ell_{41}}{d_+^3} & \dfrac{\ell_{42}}{d_+^2} & \dfrac{\ell_{43}}{d_+} & 1 & \cdots & 0 \\
  \vdots & \vdots & \vdots & \vdots & \ddots & \vdots \\
  \dfrac{\ell_{m1}}{d_+^{m-1}} & \dfrac{\ell_{m2}}{d_+^{m-2}} & \dfrac{\ell_{m3}}{d_+^{m-3}} & \dfrac{\ell_{m4}}{d_+^{m-4}} & \!\! \cdots & 1 \end{bmatrix}\,,
\end{align}
which clearly shows that
\begin{align*}
  \lim_{d_+ \rightarrow \infty} \Lplus = I\,.
\end{align*}

This means that for a sufficiently large $d_+$, $\Lplus \Lplus^T$ is diagonal dominant.
Consequently, there exists $\dmin\!>\!0$ such that
\begin{align*}
  \Lplus \Lplus^T A + A \Lplus \Lplus^T < 0\,, \quad \forall d_+ > \dmin\,,
\end{align*}
and, as a result, from  (\ref{lem3:2Q}), $Q>0$. \hfill $\blacksquare$
\end{proof}

\section{Design procedure} \label{sec:design}

The control law is chosen as
\begin{equation}\label{MLS:u}
  u = L(s) \big[ \Xi^T\Theta \big] = \Omega^T \Theta + \Xi^T \dot{\Theta}\,,
\end{equation}
where
\begin{equation}\label{MLS:L}
  L(s) = (s+\ell_0)\, I\,,
\end{equation}
with $\ell_0>0$ and $I \IN \R^{m \times m}$, and
\begin{equation}\label{MLS:Xi}
  \Xi^T = L^{-1}(s) \Omega^T\,. 
\end{equation}
%
Applying the control (\ref{MLS:u}), the error equation (\ref{MLS:e0:3}) becomes
\begin{align}
  e_0 &= M(s)L(s)SD\, \big[\, \Xi^T\Theta - \Xi^T \Theta^* \big] \notag \\
    &= M(s)L(s)SD\, \big[\, \Xi^T \tilde{\Theta}\, \big]\,, \label{MLS:e0:4}
\end{align}
where $\tilde{\Theta} = \Theta - \Theta^*$.
Now, introduce the decomposition
\begin{align*}
  M(s)L(s) &= \diag \Big\{ \frac{s+\ell_0}{s+a_i} \Big\}_{i=1}^m \\
    &= \diag \Big\{ \frac{\alpha_i}{s+a_i} + 1 \Big\}_{i=1}^m \\
    &= \mathcal{A} M(s) + I\,,
\end{align*}
where
\begin{equation*}
  \mathcal{A} = \diag \big\{ \alpha_i \big\}_{i=1}^m\,, \quad \alpha_i = \ell_0 - a_i\,.
\end{equation*}

Then, the error equation becomes
\begin{align}
  e_0 &= \big( \mathcal{A} M(s) + I \big) SD \big[ \Xi^T \tilde{\Theta} \big] \notag \\
    &= \A M(s)\, SD \big[ \Xi^T \tilde{\Theta} \big] + SD \big[ \Xi^T \tilde{\Theta} \big]\,. \label{MLS:e0:5}
\end{align}
For simplicity, let us denote $\Ups$ the filtered control error ($DL^{-1} \tilde{u}$), i.e.,
\begin{equation}\label{eq:ups}
 \Ups = D \big[ \Xi^T \tilde{\Theta} \big]\,. 
\end{equation}
The total order of the system comprising the plant (\ref{MLS:y}) and filters (\ref{MLS:dv1})-(\ref{MLS:dv2}) and (\ref{MLS:Xi}) is $N = n + 4m\nu - m - 1$.
Then, defining the error state vector $e\IN\R^{N}$, we can write the following non--minimal state space realization of (\ref{MLS:e0:5})
\begin{align}
  \dot{e} &= A_c e + B_c  S \Ups\,,     \label{MLS:de} \\
  e_0 &= \A C_c e +  S \Ups\,,  \label{MLS:Ae0:2}
\end{align}
where $M(s) = C_c (sI - A_c)^{-1}B_c$,  and $\{A_c, B_c, C_c\}$ satisfies the MKY Lemma (appropriate for nonminimal state-space realizations):
\begin{align}
  A_c^TP' + P'A_c &= -2Q'\,,  \label{MLS:MKY:1} \\
  P'B_c &= C_c^T\,,           \label{MLS:MKY:2}
\end{align}
with $P'=P'^T>0$ and $Q'=Q'^T>0$.
\begin{rem}\label{rem1}
 Note that we need to filter only $\Omega_1$ to form $\Xi$.   
\end{rem} 
%

\begin{table}[!htb]
  \centering
  \renewcommand{\arraystretch}{1.3}
  \begin{tabular}{|l|ll|}
  \hline
  Tracking error 
    & $e_0 = y - y_m$ \\
  \hline
  SV filters  
    & $\dot{v}_{1,i} = \Lambda v_{1,i} + g u_i$ &  (\ref{MLS:dv1}) \\
    & $\dot{v}_{2,i} = \Lambda v_{2,i} + g y_i$ &  (\ref{MLS:dv2}) \\
    & $v_1^T = \big[v_{1,1}^T \ \ v_{1,2}^T \ \ \cdots \ \ v_{1,m}^T \big]$ & \\
    & $v_2^T = \big[v_{2,1}^T \ \ v_{2,2}^T \ \ \cdots \ \ v_{2,m}^T \big]$ & \\
    & $\omega^T = \big[v_1^T \ \ y^T \ \ v_2^T \ \ r^T \big]$ & \\[1mm]
    & $\Omega_i^T = [\omega^T\quad u_{i+1} \ \cdots \ u_m]$ & (\ref{MLS:Omegai}) \\[1mm]
    & $\displaystyle \Omega^T = \diag \big\{ \Omega_i^T \big\}_{i=1}^m$ & (\ref{MLS:Omega}) \\[0.5mm]
  \hline
  $\Xi$-filter 
    & $\dot{\Xi} = -\ell_0\,\Xi + \Omega$ \UP & (\ref{MLS:Xi}) \\
    & $\ell_0= a \,, \quad a_i\equiv a, \quad \forall i \IN [1,m]$ & \\
  \hline
  Control 
    & $u = \Omega^T \Theta + \Xi^T \dot{\Theta}$ \UP & (\ref{MLS:u}) \\
  \hline
  Update laws 
    & $\displaystyle \dot{\Theta} = -\gamma\, R\, \Xi\,  \sign(D)\, e_0$  \UP & \\
    & $\dot{R} = -R\, \Xi\, \Xi^T R$ & \\
    & $R(0) = \diag\big\{ R_i(0) \big\}_{i=1}^m$ & \\[1mm]
    & $R_i(0) = R_i^T(0) >0$ & \\
    & $R_i(0) \IN \R^{N_i \times N_i}$ & \\
  \hline
  \end{tabular}
  \renewcommand{\arraystretch}{1}
  \caption{MIMO LS-MRAC algorithm.}\label{MIMO LS-MRAC table}
\end{table}

\subsection{Fast converging LS-MRAC}\label{sec:fastLSMRAC}
In this section we consider a particular case of the design characterized by a reference model given by a scaled identity matrix:
\begin{align}\label{eq:Am_identity}
A_m=-a I\,.
\end{align}

In this case ($\A=0$), we have the following non--minimal state space realization of (\ref{MLS:e0:5})

\begin{align}
  \dot{e} &= A_c e + B_c S \Ups\,, 
  \label{eq:defast} \\
    e_0 &=  S \Ups\,. 
    \label{eq:e0fast}
\end{align}

Obviously, the state $e$ is not observable from the tracking error $e_0$ but remains important to consider for overall stability analysis. Equivalently, one has 
\begin{align}
  \dot{e} &= A_c e + B_c e_0 \,,     \label{MLS:de2} \\
    e_0 &=  S \Ups \,.  \label{eq:e0fast2}
\end{align}
Now define the matrix
\begin{align}\label{MLS:D}
  \D &= \diag\big\{ d_i I_i 
  \big\}_{i=1}^m\,,
\end{align}
where $d_i$ are the diagonal elements of matrix $D$, $I_i \IN \R^{N_i \times N_i}$, $N_i = (2m\nu+m-i)$. Then, consider the time varying positive definite function
\begin{equation}\label{eq:V3}
  2V_1(\tilde{\Theta},t) =  \tilde{\Theta}^T |\D|\, R^{-1}(t)\, \tilde{\Theta}\,, 
\end{equation}
where {the covariance matrix is given by}
\begin{align*}
  R(t) &= \diag\big\{ R_i(t) \big\}_{i=1}^m\,, \\
  R_i(0) &= R_i(0)^T>0\,, \quad R_i(0) \IN \R^{N_i \times N_i}\,.
\end{align*}

%
%
Using the fact that $\dot{R}^{-1} = - R^{-1} \dot{R} R^{-1}$, then
\begin{align}\label{eq:dV3}
  \dot{V}_1 &=  \tilde{\Theta}^T |\D| R^{-1} \dot{\Theta}
    - \frac{1}{2} \tilde{\Theta}^T |\D| R^{-1} \dot{R} R^{-1} \tilde{\Theta}\,.
\end{align}

In view of the above equation, we choose the update law for  $R$ given in Table~\ref{MIMO LS-MRAC table}. 
Since $\Xi$ is block diagonal, then $\Xi\, \Xi^T$ is also block diagonal.
Hence, for $R(0)$ block diagonal, $\dot{R}$ and $R$ result block diagonal and also
$\dot{R}^{-1} = \Xi\, \Xi^T$. 
Thus, 
\begin{align}
  \dot{V}_1 &= \tilde{\Theta}^T |\D| R^{-1} \dot{\Theta}
    + \frac{1}{2} \tilde{\Theta}^T |\D| \Xi \Xi^T  \tilde{\Theta}\,,
\end{align}
which, upon setting the update law
\begin{align}
  \dot{\Theta} &= -\gamma\, R\, \sign(\D) \Xi e_0 \notag \\ 
    &= -\gamma\, R\, \Xi\, \sign(D) e_0 \label{eq:dTheta}
\end{align}
is reduced \textcolor{black}{using  (\ref{eq:ups})} to (note that $D^{-1}|D|D^{-1}=|D^{-1}|$)
\begin{align}
  \dot{V}_1 &=
  - \gamma \Ups^T\, \big( S -\frac{1}{2\gamma} |D^{-1}|\big)\, \Ups\,.  \label{eq:dV3a}
\end{align}
Then, for sufficiently large $\gamma>0$, we get
\begin{equation}\label{eq:gammalarge}
S -\frac{1}{2\gamma}|D^{-1}| \geq Q_4 >0\,,
\end{equation}
and
\begin{align}
    \dot {V}_1 &\leq  -\gamma \Ups^T Q_4 \Ups \label{eq:dV1-3}\,.
\end{align}
Next, we take into account the error dynamic equation (\ref{eq:defast}) by introducing the positive definite function of $e$
\begin{align}\label{eq:V1e}
     2 {V}_2 &= - e^T P_c e\,, 
\end{align}
such that $A^T_c P_c +P_c A_c=-Q_c<0$, for some $P_c, Q_c >0$,  which exist since $A_c$ is Hurwitz. Then, consider the Lyapunov function for the adaptive system (\ref{eq:defast}), (\ref{eq:e0fast}), and (\ref{eq:dTheta})
\begin{align}\label{eq:Vfast}
    V = & \varepsilon V_2(e)+V_1(\tilde{\Theta},t)\,, \quad \varepsilon > 0\,,
\end{align}
which has time-derivative given by
\begin{align}\label{eq:dVfast}
    \dot{V} &\leq - \begin{bmatrix} e^T  &  \Ups^T \end{bmatrix} 
    \begin{bmatrix}
        \varepsilon Q_c  &  -\varepsilon P_c B_c S \\
        0     &  \gamma Q_4
    \end{bmatrix}
    \begin{bmatrix} e \\ \Ups 
    \end{bmatrix}\,,
\end{align}
which is negative definite in $\begin{bmatrix} e^T  &  \Ups^T \end{bmatrix}$ if $\varepsilon$ is small enough.

Such condition can be simply expressed in terms of the original LDU decomposition of $K_p$. This can be concluded from the $D_+$ parameterization $S=L_pD_+ L_p^T$ and from the fact that $L_+$ (\ref{MLS:L+}) is close to the identity for large $d_+$. Indeed, (\ref{eq:gammalarge}) holds if 
\begin{equation}\label{eq:gammalarge1}
S=L_p D_+L_p^T > \frac{1}{2\gamma}|D^{-1}|=\frac{1}{2\gamma}|D_p^{-1}|\ D_+ 
\end{equation} 
or, equivalently  
\begin{equation}\label{eq:L+L+T}
L_+ L_+^T > \frac{1}{2\gamma}|D_p^{-1}|. 
\end{equation}
Since there exists $d_+$ such that $L_+$ becomes arbitrarily close to the identity, the following theorem can be proved.

\begin{theo}\label{theo1}
For the closed-loop system consisting of the plant (\ref{MLS:y}), reference model (\ref{MLS:ym}) (\ref{eq:Am_identity}), and the MIMO LS-MRAC algorithm summarized in Table \ref{MIMO LS-MRAC table}, where $\ell_0=a$, if
\begin{align}\label{Th2:gamma}
  \gamma &> \frac{1}{2} \max_i \big( |d_{p_i}|^{-1} \big)\,, \quad i \IN [1,m]\,,
\end{align}
where $\{d_{p_i}\}_{i=1}^m$ are the diagonal elements of matrix $D_p$, then: {\bf(a)} all system  signals are globally uniformly bounded and $e(t), e_0(t),  \tilde{\Theta}^T\Xi(t) \in \Ltwo \cap \Linf$ and tend to $0$ as $t \rightarrow \infty$; {\bf(b)} if $R(0) \geq c\, \gamma\, I$ for some constant $c>0$, then \textcolor{black}{the norms} $\|e_0\|^2_2,\  \|\Ups\|^2_2,  \|e-e_h\|^2_2$ 
\textcolor{black}{are of order}
$\O(\gamma^{-2})\,, $ where $e_h$ is the homogeneous exponentially decaying solution of the linear subsystem (\ref{eq:defast}).
\end{theo}

\begin{proof}  From (\ref{eq:gammalarge}) (\ref{Th2:gamma}), $\dot V \leq 0$. Then,  we conclude that $V$ is bounded. 
To show the boundedness of $e(t), \tilde{\Theta}(t)$ we must take into account that $V_1$ (as a function of $\tilde{\Theta}$) is time-varying due to the term $R^{-1}(t)$. %
To this end, note that since $R(0) = R^T(0) >0$, then from the $R$ update law
\begin{equation*}
  \dot{R}^{-1}(t) = \Xi\Xi^T\,,
\end{equation*}
which by integrating gives
\begin{equation}\label{Th2:inv Rm}
  R^{-1}(t) = R^{-1}(0) + J(t) \geq 0\,, \quad t \geq 0\,,
\end{equation}
where $ J(t) = \int_0^t \Xi(\tau)\Xi^T(\tau) d\tau.$
Now, from (\ref{eq:Vfast}) and (\ref{Th2:inv Rm}),
\begin{align*}
  V(e, \tilde{\Theta}) &= \varepsilon V_2(e)
    + \tilde{\Theta}^T |\D| R^{-1}(0) \tilde{\Theta}
    + \tilde{\Theta}^T |\D| J(t) \tilde{\Theta}\,.
\end{align*}

 Since all the terms in the above sum are nonnegative, they must be bounded. Consequently, $e,\ \tilde{\Theta }, \ \Theta \IN \Linf$. Since $\dot{V}$ is bounded from above by a negative definite quadratic form of $e,\Ups$, then  $e,\  \Ups \IN \Ltwo$.

 Considering \textcolor{black}{the $\Omega$ structure} in (\ref{MLS:Omegai}), we easily conclude subsequently from 
 $(\Omega_m \IN \Linf)$ that $(u_m \IN \Linf)$ $\Rightarrow$
  $(\Omega_{m-1} \IN \Linf)$ $\Rightarrow$ $(u_{m-1} \IN \Linf)$ $\Rightarrow$
  $\dots$ $\Rightarrow$ $(u_2 \IN \Linf)$ $\Rightarrow$
  $(\Omega_1 \IN \Linf)$ $\Rightarrow$ $(u_1 \IN \Linf)$.
This allows us to conclude that $\Omega, \ \Xi,\ \dot{\Xi} \IN \Linf$, and $u \IN \Linf$.
Therefore, all system signals are globally uniformly bounded.

Moreover, since $e, \tilde{\Theta}^T\Xi \IN \Ltwo \cap \Linf$ and $\dot{e}, \frac{d}{dt}\tilde{\Theta}^T\Xi \IN \Linf$ then, from Barb\v{a}lat Lemma \cite[p. 81]{Tao:2003}, it follows that $\lim_{t \rightarrow \infty} e = 0$ and $\lim_{t \rightarrow \infty} \tilde{\Theta}^T\Xi = 0$.

It is also clear that $\Ups \in \Ltwo \cap \Linf$ . Indeed, $V_1$ is bounded by its initial value since $\dot{V}_1\leq 0$. Moreover, by integrating both extreme sides of (\ref{eq:dV1-3}), $\|\Ups\|^2_{2} = \O(\gamma^{-2})$. Note that $V_1(0) = \O(\gamma^{-1})$ by the $R(0)$ assumption in {\bf (b)}. Since in (\ref{eq:defast}) $A_c$ is Hurwitz and $e_0$ is given by (\ref{eq:e0fast}), then property {\bf(b)}  holds.
\end{proof}

\subsection{Summary of the MIMO LS-MRAC}

Table \ref{MIMO LS-MRAC table} summarizes the MIMO LS-MRAC algorithm.

\begin{rem}\label{rem2}
The control law (\ref{MLS:u})  can be written as
\begin{align*}
  u 
    &= \Theta^T \Omega - \gamma e^T_0 \sign(D) \Xi^T R \,\Xi \,,
\end{align*}
which shows a feedback term of the tracking error $e_0$.
This seems related to the proposed algorithm fast  tracking error convergence.

\end{rem}

\subsubsection {The M-MRAC}\label{sssec:M-MRAC}
The M-MRAC algorithm described in \cite{Costa:2020} is a particular case of the LS-MRAC.
It is easily obtained by using a constant adaptation gain $\Gamma$ setting
\begin{align*}
  \dot{R} &= 0\,, \\
  R(0) &= \gamma^{-1} \diag\big\{ \Gamma_i \big\}_{i=1}^{m}\,, \\
  \Gamma_i &= \Gamma_i^T >0\,, \quad \Gamma_i \IN \R^{N_i \times N_i}\,.
\end{align*}
The M-MRAC reduces to the MRAC with gradient adaptive law applied to a relative degree zero {\em output} error equation. Thus, Remark~\ref{rem2}
applies in this case and a fast convergence of the tracking error  results. However, as illustrated in Section~\ref{sec:simulations}, the parametric convergence can be very slow.

The foregoing LS-MRAC stability analysis can be easily extended to the M-MRAC case by just setting $\dot{R}=0$. In this case, inequality (\ref{eq:gammalarge}) becomes simply $\gamma S>0$, i.e., it suffices that $\gamma$ be positive, for stability of the M-MRAC.
\subsection{Convergence analysis}\label{sec:convergence}
\textcolor{black}{First, let us present the following lemmas.} 

\begin{lem}\label{lem3}
  Given the inequality
\begin{align}\label{eq:lem3_1}
  \dot{V}_1(t) &\leq -\gamma e_1^2(t)\,,
\end{align}
where $V_1(t)\geq 0\ \forall t$, $e_1(t)$ is continuous, and $\gamma>0$.
For any given $\varepsilon>0$, such that
\begin{align}\label{eq:lem3_2}
 e^2_1(t) &\geq {\varepsilon}^2, \quad \forall t\in [t_0, t_0 +\Delta_{\varepsilon}] , 
 \end {align}
then, if $V_1(0)=c_0 \gamma^{-1}$, the time interval $\Delta_{\epsilon}$ must satisfy, 
\begin{align}\label{eq:lem3_3}
\Delta_{\varepsilon}\leq  c_0 (\varepsilon\gamma)^{-2}\,.
 \end{align}
\end{lem}
\begin{proof}
  Suppose that initially $e^2_1(t_0)\geq \varepsilon^2$. By continuity,  (\ref{eq:lem3_2}) follows for some $\Delta_\varepsilon \geq 0$. Then, from (\ref{eq:lem3_1}) we have 
\begin{align}\label{eq:lem3_4}
0 \leq  V_1(t) &\leq V_1(0)-\gamma \varepsilon^2 \Delta_\varepsilon
\end{align}
and thus,
\begin{align}\label{eq:lem3_5}
 \Delta_\varepsilon \leq c_0(\varepsilon\gamma)^{-2}\,.
\end{align}
\end{proof}
The following result relates  $\Ltwo$-norm and $\Linf$-norm of uniformly Lipschitz-continuous functions.

\begin{lem} \label{lem4}
Let  $f(t)$ be uniformly continuous with Lipschitz constant $K$. Then, for any positive $\varepsilon$
\begin{equation}
    \|f(t)\|^2_2 = \varepsilon \ \ \Rightarrow \ \ \|f(t)\|_{\infty}\leq (3K  \varepsilon)^{1/3}.
\end{equation}
\end{lem}
\begin{proof}
Given the Lipschitz constant $K$, we have $|f(t)|\geq |f(t_0)|- K(t-t_0), \ \  \forall t, t_0$. Let $t_m=|f(t_0)|/K +t_0$ be the time when the right-hand-side of the latter inequality is zero. Then,
\begin{align}\label{eq:lem5_1}
 \varepsilon&=\|f\|^2_2 \geq \int_{t_0}^{t_m} (|f(t_0)|-K(t-t_0))^2 dt \notag \\
& = -(|f(t_0)|-K(t-t_0))^3/3K|^{t_m}_{t_0}= |f(t_0)|^3/3K.
\end{align}

Since, this holds for any $t_0$, 
\begin{align}\label{eq:lem5_2}
 |f(t)|\leq (3K\varepsilon)^{1/3}, \ \forall t.
\end{align}
This proves the Lemma.
\end{proof}
Now we can prove the convergence theorem.

\begin{theo}[\bf Convergence Characterization]\label{theo2}
    In the LS-MRAC system of Theorem~\ref{theo1}, if the $R(0) \geq c\, \gamma\, I$ for some constant $c>0$, then the tracking error converges to an arbitrarily small residual value $\varepsilon$ within a time of order  $\O(\gamma^{-1})$, after that it remains norm-bounded by a constant of order  $\O(\gamma^{-1/2})$ and ultimately converges to zero.
\end{theo}

\begin{proof}
From (\ref{eq:e0fast}) we have 
\begin{align}\label{eq:e0dot}
 \dot{e}_0=S \dot{\Ups}=D [\dot{\Xi}^T \tilde{\Theta} + \Xi^T \dot{\tilde{\Theta}}]\,.
\end{align}
Since $\Xi, \dot{\Xi}^T \tilde{\Theta}, R, \dot{R}$ are all uniformly bounded, from (\ref{eq:dTheta}) we can write, for some constant $c_1>0$
\begin{align}\label{eq:bound_Thetadot}
  |\dot{\tilde{\Theta}}| \leq \gamma c_1 \|e_0\|_{\infty}\,.
\end{align}
Then, from (\ref{eq:e0dot}), the Lipschitz constant satisfies
$$K\leq \gamma c_2\|e_0\|_{\infty},$$
 for some $c_2>0$ and large enough $\gamma$.

From (\ref{eq:dV3a}) and (\ref{eq:V3}), noting that $R(0) \geq c\, \gamma\, I$  
    $$\|e_0\|^2_2\leq c_3/(\gamma)^2, \quad \exists c_3>0.$$ 
Now, from Lemma~\ref{lem4}, one has 
\begin{align}\nonumber
   \|e_0\|_{\infty} &\leq (3\gamma c_2 \|e_0\|_{\infty|})^{1/3} \|e_0\|_2 \\
            &\leq (3\gamma c_2 \|e_0\|_{\infty|})^{1/3} c_4 \gamma^{-1/3}\,.\label{eq:e0inftyimpl}
\end{align} 
Now, solving (\ref{eq:e0inftyimpl}) we get
\begin{align}\label{eq:e0inftyboundl}
   \|e_0\|_{\infty} \leq  c_5 \gamma^{-1/2}.     
\end{align}
Moreover, using Lemma~\ref{lem4}, after a time interval of at most $\Delta_{\varepsilon}$, the tracking error reaches the residual error given by (\ref{eq:e0inftyboundl}).
\end{proof}
    
\begin{rem}[The case of a general diagonal $A_m$ (\ref{eq:Am})]\label{rem3}
The convergence properties of the general case was verified to be similar to the particular case analysed above. A full analysis seems possible but more involved and is not pursued here. However, note that tracking a  more general model reference can be achieved with the particular model of Theorem~\ref{theo1}. Indeed, it is sufficient to filter the reference signal with  a simple first order filter given by 
$$ F(s) = \diag \left\{ \frac{s+a}{s+a_i}\right\}^m_1.$$
\end{rem}
\section{Simulation results}\label{sec:simulations}

To illustrate the improvement introduced by the presented MIMO algorithms two plants are considered \cite{HCIK:2001}.
The first example corresponds to the simplest possible 2I2O adaptive control problem that appears in visual servoing with an uncalibrated camera.
The system is described by
\begin{align}
  P(s) &= \diag \left\{ \frac{1}{s+2}, \frac{1}{s+2} \right\} K_p\,, \label{P1} \\
  K_p &= \begin{bmatrix} \cos \phi & \sin \phi \\
     -h \sin \phi & h \cos \phi \end{bmatrix}\,, \\
  M(s) &= \diag \Big\{ \frac{2}{s+2}, \frac{2}{s+2} \Big\}\,,
\end{align}
where the only unknown parameters are $\phi$ and $h$, respctively, the camera misalignment and the scaling factor of the image.
The simulation results are obtained with $\phi=1$ and $h=0.5$.
The corresponding matching parameters for this system are
  $\Theta_1^{*T} = [0.54 \ \ -1.68 \ \ 0]$, \
  $\Theta_2^{*T} = [0.84 \ \ 1.08]$, \
  $\|\Theta^*\| = 2.34$.

The second example is a 2I2O plant of third order.
The simulations are performed with
\begin{align}
  P(s) &= \frac{1}{s^2-1} \begin{bmatrix} s+3 & 2s \\
         -2s-4 & s+3 \end{bmatrix} \,, \label{P2} \\
  M(s) &= \diag \left\{ \frac{1}{s+1} , \frac{2}{s+2}\right\}\,.
\end{align}

This plant has poles at $s=\{1,1,-1\}$, transmission zero at $s=-1.8$, and
\begin{equation*}
  K_p = \begin{bmatrix} 1 & 2 \\ -2 & 1 \end{bmatrix}\,,
\end{equation*}

and state-space realization $$A_p = \begin{bmatrix} 1 & 0 &0\\ 0& 1&0\\0&0 &-1\end{bmatrix}; \ B_p=\begin{bmatrix}1 & 1 \\ 1& 0\\ 1 & -1\end{bmatrix}; \ C_p=\begin{bmatrix} 1 & 1 & -1 \\2 & -5 & 1\end{bmatrix}.$$
This case requires 17 parameters to be adapted, that is, $\Omega_1 \in \R^{9}$, and $\Omega_2 \in \R^{8}$.

Plant (\ref{P1}) is considered for simulations 1 to 3 with
\begin{align*}
  y_p(0) &= [1 \ \ 1]^T\,, \qquad y_M(0) = [0 \ \ 0]^T\,, \\
  r(t) &= [1+10\sin(5t) \ \ \ -1+5\sin(3t)]^T\,.
\end{align*}
All initial conditions not mentioned are set to zero.

\paragraph{Simulation 1.}\label{MIMO MRAC simulations}

For comparison, Fig. \ref{Fig1} shows the simulation results of the plant (\ref{P1}) with the MIMO MRAC algorithm of \cite{CHIK:03}.
The following gain is set:
\begin{align*}
  \Gamma = 10\,I\,.
\end{align*}

The large transient tracking error 
shown in Fig. \ref{Fig1} would hardly be acceptable in practice.
Increasing the gain $\Gamma$ in this algorithm does not help.

\begin{figure}[!htb]
  \centering
  \parbox[c]{3mm}{a)}\hfill
  \parbox[c]{80mm}{
    \includegraphics[width=7.0cm]{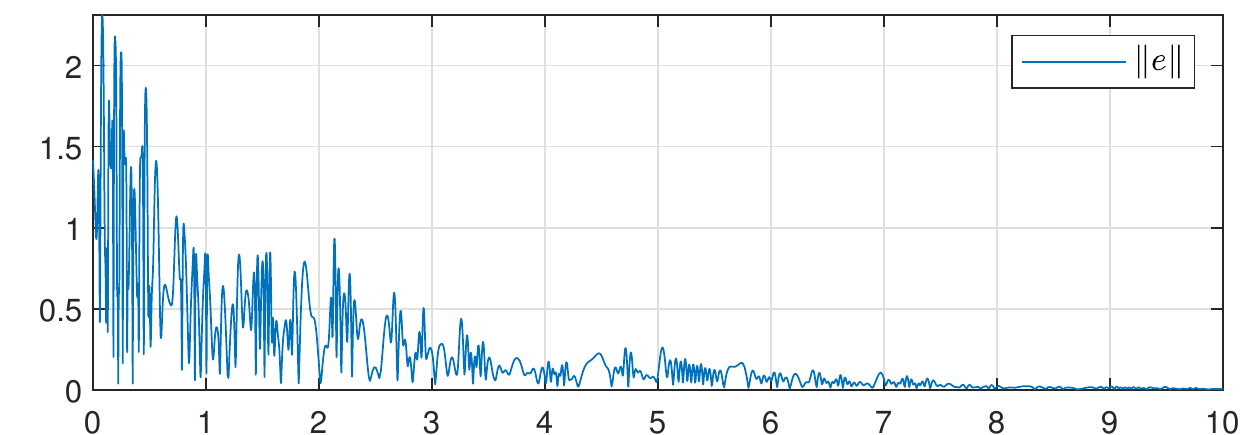}    
  } \\[3mm]
  \parbox[c]{3mm}{b)}\hfill
  \parbox[c]{80mm}{
    \includegraphics[width=7.0cm]{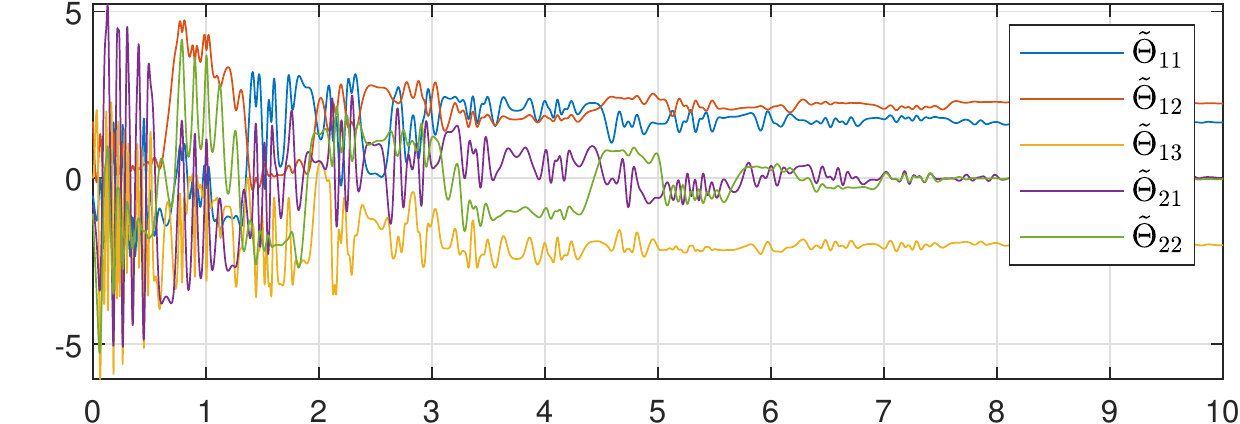}   
  }
  \caption{Simulation result of the first order plant (\ref{P1}) with the MIMO MRAC algorithm.}
  \label{Fig1}
\end{figure}

\paragraph{Simulation 2.}\label{MIMO M-MRAC simulation}

Figure \ref{Fig2} shows the simulation results of the plant (\ref{P1}) with the MIMO M-MRAC algorithm summarized in Table \ref{MIMO LS-MRAC table} and Section \ref{sssec:M-MRAC}.
The following data are used:
\begin{align*}
  \ell_0 &= 3\,, \quad \Gamma = 500\,I\,.
\end{align*}

The tracking error transient and the control mismatch are remarkably improved.
Here we can verify the effect of a large adaptation gain $\Gamma$, as pointed out in Remarks \ref{rem2} and \ref{rem4}.
Large $\Gamma$ reduces the tracking error and the control mismatch, but also reduces the convergence rate of the parameter (due to small $e_0$).
%


\begin{figure}[!htb]
  \centering
  \parbox[c]{3mm}{a)}\hfill
  \parbox[c]{80mm}{
    \includegraphics[width=7.0cm]{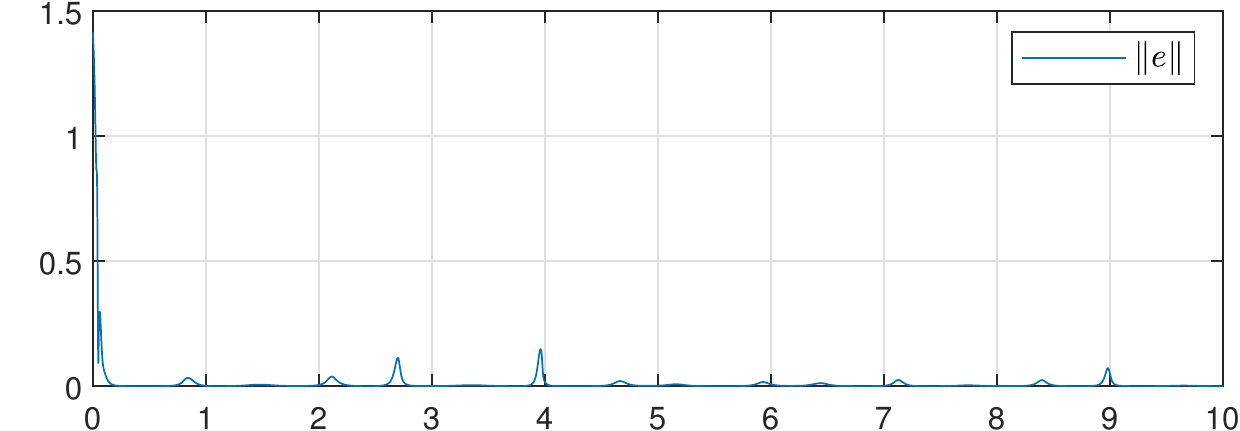}    
  } \\[3mm]
  \parbox[c]{3mm}{b)}\hfill
  \parbox[c]{80mm}{
    \includegraphics[width=7.0cm]{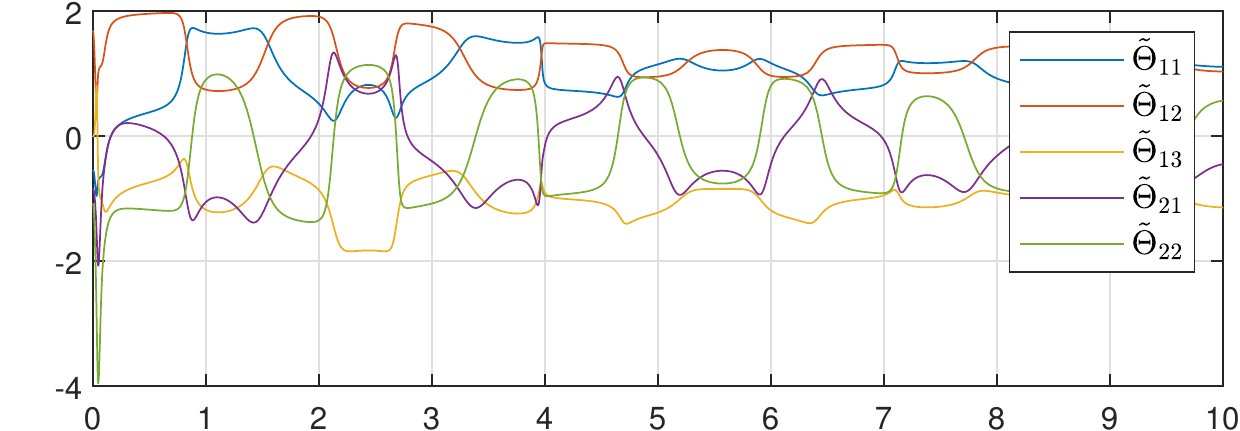}    
  }
  \caption{Simulation result of the first order plant (\ref{P1}) with the MIMO M-MRAC algorithm.}
  \label{Fig2}
\end{figure}

\paragraph{Simulation 3.}\label{MIMO LS-MRAC simulation3}
Figure \ref{Fig3} shows the simulation results of the plant (\ref{P1}) with the MIMO LS-MRAC algorithm summarized in Table \ref{MIMO LS-MRAC table}.
The following data are used:
\begin{align*}
  \ell_0 &= 3\,, \quad \gamma = 50\,, \quad R(0) = 20\,I\,.
\end{align*}

Note the fast and smooth convergence of the parameter even when employing large gains.


\begin{figure}[!htb]
  \centering
  \parbox[c]{3mm}{a)}\hfill
  \parbox[c]{80mm}{
    \includegraphics[width=7.0cm]{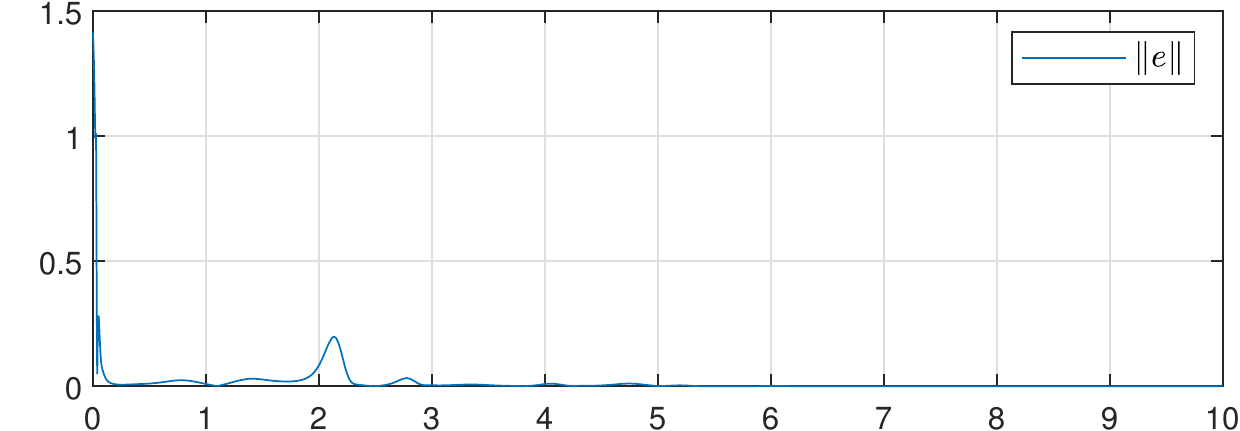}    
  } \\[3mm]
  \parbox[c]{3mm}{b)}\hfill
  \parbox[c]{80mm}{
    \includegraphics[width=7.0cm]{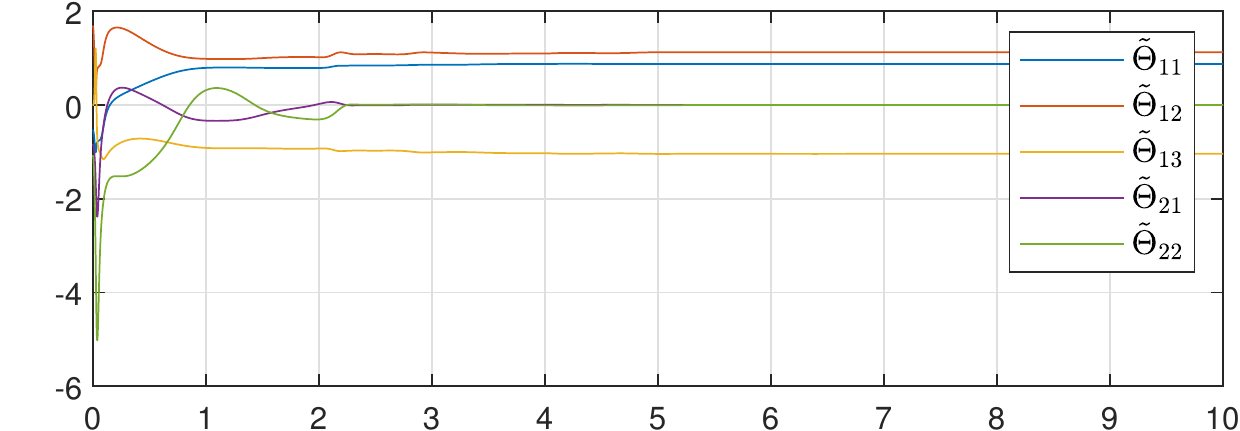}    
  }
  \caption{Simulation result of the first order plant (\ref{P1}) with the MIMO LS-MRAC algorithm.}
  \label{Fig3}
\end{figure}


The next simulations refer to the 3rd order plant (\ref{P2}).

\paragraph{Simulation 4.}\label{MIMO LS-MRAC simulation4}

Figures \ref{Fig4} to \ref{Fig8} show the simulation results of the plant (\ref{P2}) with the MIMO LS-MRAC algorithm.
This is a case with 17 parameters where 
\begin{align*}
  &\nu = 2\,, \quad \ell_0 = 2\,, \quad \Lambda = -2\,I\,, \quad g = I\,, \quad
 A_M=-2\,I\,,\\
%
%
%
  &r(t) = [1+10\sign(\sin(5t) )\ \ \ -1+5 \sign(\sin(3t))]^T\,.
\end{align*}
All initial conditions not mentioned are zero.

%
\begin{figure}[!htb]
  \centering
  \parbox[c]{3mm}{ }\hfill
  \parbox[c]{80mm}{
    \includegraphics[width=7.0cm]{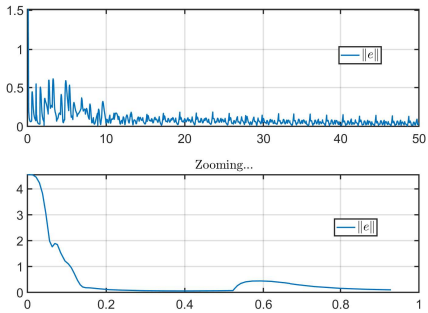} 
  }

  \caption{Simulation result of the third order plant (\ref{P2}) with the MIMO LS-MRAC algorithm, $x_p(0) = [0.65\ \ 1\ \ -0.37]^T$, $\gamma=10$ and $R(0)=I$.}
  \label{Fig4}
\end{figure}
%

\begin{figure}[!htb]
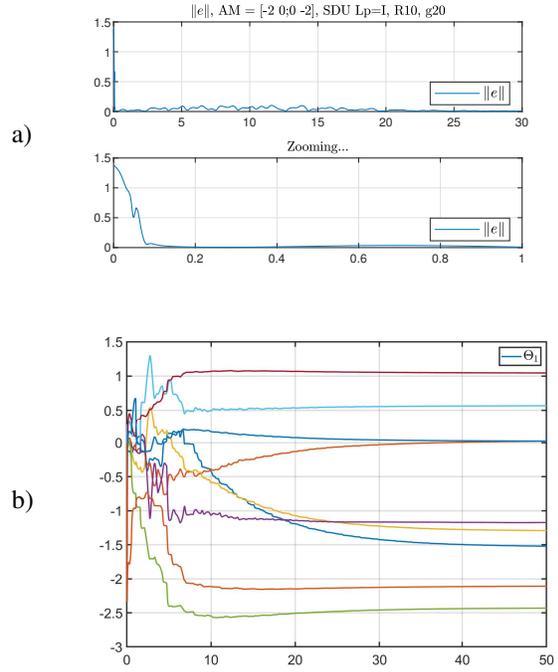

  \centering
  \parbox[c]{3mm}{a)}\hfill
  \parbox[c]{80mm}{
    \includegraphics[width=7.0cm]{MIMO_M_MRAC_2317_Fig5_e+zoom.eps} 
 }\\[3mm]   
  \parbox[c]{3mm}{b)}\hfill
  \parbox[c]{80mm}{
     \includegraphics[width=7.5cm]{MIMO_LS_MRAC_2317_Fig4_2_theta1_rev.pdf} 
 } 
 
 \caption{As in Fig.~\ref{Fig4}, with larger R(0)=10. Note that the tracking transient is faster and with smaller errors, as expected. The parameters converge nicely, as shown for $\Theta_1$.}
  \label{Fig5}
\end{figure}


\begin{figure}[!htb]
  \centering
  \parbox[c]{3mm}{~~}\hfill
  \parbox[c]{80mm}{
    \includegraphics[width=7.0cm]{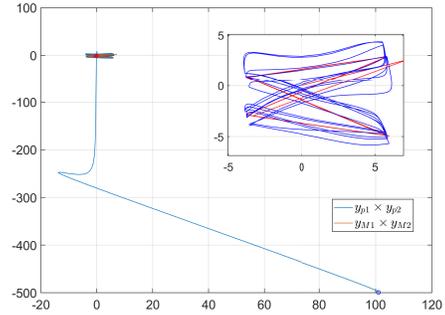} 
 }
 \caption{As in Fig.~\ref{Fig4}, with large initial conditions, $x_p(0) = [0.65\ \ 100\ \ -0.37]^T$, $\gamma=10$, $R(0)=20$. Large residual errors result due to large transients in the adaptive parameters. The subpicture displays the residual set around the reference trajectory consisting of a poligonal Lissajoux pattern (in red).}
  \label{Fig6}
\end{figure}


\begin{figure}[!htb]
  \centering
  \parbox[c]{3mm}{~~}\hfill
  \parbox[c]{80mm}{
    \includegraphics[width=7.0cm]{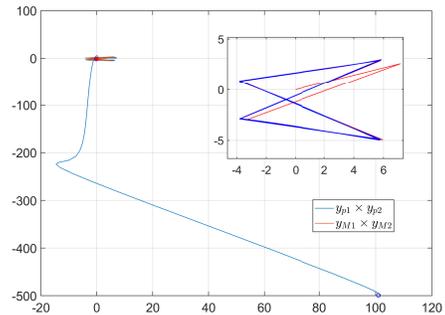} 
 }
 \caption{As in Fig.~\ref{Fig6}, with large initial conditions but with switched $\sigma$-modification ($0$ to $10$). Small residual errors result due to reduced transients in the adaptive parameters.}
  \label{Fig7}
\end{figure}

%
\begin{figure}[!htb]
  \centering
  \parbox[c]{3mm}{~~}\hfill
  \parbox[c]{80mm}{
    \includegraphics[width=7.0cm]{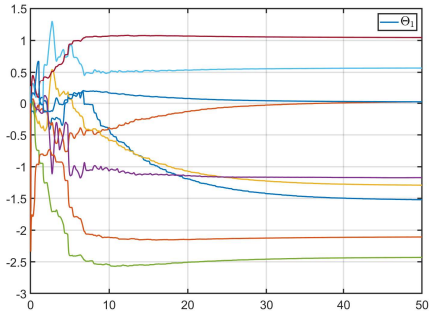} 
 }
 \caption{Parameter convergence ($\Theta_1$) for the case of Figure\ref{Fig7} with  initial conditions $x_p(0) = [0.65\ \ 1\ \ -0.37]^T$. Similar behavior occurs with $\Theta_2$.}
  \label{Fig8}
\end{figure}


\subsection{On Parametric Convergence}
Since $A_c$ is Hurwitz, the state error asymptotically tends to a small $\Ltwo$ signal. This implies that the regressor $\xi(t)$ is asymptotically persistently exciting (PE) if the reference signal is rich and $\mathcal{O}(\gamma^{-1})$ is sufficiently small. Consequently, the parameters converge to the matching values with the regressor approaching the model regressor $\xi_m$, thus, inheriting the favorable convergence properties for the least-squares (LS) identification of the linear regression problem (\ref{MLS:e0}) with $\xi=\xi_m$. This is illustrated in Fig.~\ref{Fig3}b.

Figures \ref{Fig6} and \ref{Fig7} illustrate the transient behavior for very large initial conditions in the plant state. In Fig.~\ref{Fig6}, a rather large residual error persists. A possible explanation is that the transient excursions of some adaptive parameters is quite large, causing the residual error to remain significant for a longer period. This is predicted by Lemma 4, as the Lipschitz constant $K$ becomes large during the transient phase, and the PE property of the reference regressor may be weakened. This poor transient and steady-state residual can be easily avoided by introducing a switched $\sigma$-modification \cite{IS:96} to reduce the initial transient of the adaptive parameters. Consequently, the residual error almost vanishes, as shown in Fig.~\ref{Fig7}.
\begin{rem}\label{rem4} A similar tracking behavior as in Figs.~\ref{Fig4} or \ref{Fig5} occurs with the M-MRAC during the initial transient in Fig.\ref{Fig2} showing a fast tracking error convergence to a small residual set. However, the parametric adaptation will follow the gradient identification process rather than the faster LS identification process. In contrast to  Fig.~\ref{Fig3} (LS-MRAC), Figure~\ref{Fig2} (M-MRAC shows that the tracking is precise but parameter convergence is far from being reached.
\end{rem}

\section{Conclusion}\label{Conclusion}
A direct least-squares model-reference adaptive control (LS-MRAC) is proposed for MIMO systems with uniform relative degree one, accompanied by a complete Lyapunov-based stability analysis and convergence characterization. To the best of our knowledge, this represents the first such solution in the adaptive control literature. The MIMO LS-MRAC  achieves, in ideal conditions,  arbitrarily fast tracking error convergence while maintaining parameter convergence for appropriate adaptation gains. Simulations are presented to illustrate  the significant improvement in both tracking and parameter convergence behavior. LS-adaptation with covariance reset and forgetting factor, as well as recently proposed LS identification methods \cite{Pan_Shi_Ortega:2024} for superior performance and robustness are promising topics for future research. The extension of the algorithm to the case of  general relative degree is a significant and challenging open problem.






\bibliographystyle{ieeetr} 
\bibliography{MIMO-LS-MRAC.bib,GDiffRefs-obcat.bib,krstic_automatica2009,fidan_copilot}

%

\end{document}